\documentclass[a4paper]{lipics-v2019}

\usepackage[utf8]{inputenc}

\usepackage{enumerate,todonotes}

\title{Symmetric Promise Constraint Satisfaction Problems: Beyond the Boolean Case}

\author{Libor Barto}{Department of Algebra, Faculty of Mathematics and Physics, Charles University, Czechia }{libor.barto@gmail.com}{https://orcid.org/0000-0002-1825-0097}{}

\author{Diego Battistelli}
{Department of Algebra, Faculty of Mathematics and Physics, Charles University, Czechia}
{diego\_ew@yahoo.it}
{}
{}

\author{Kevin M. Berg}
{Department of Algebra, Faculty of Mathematics and Physics, Charles University, Czechia}
{berg.kevinm@gmail.com}
{https://orcid.org/0000-0002-1555-4239}
{}

\authorrunning{L. Barto, D. Battistelli, K.\,M. Berg}

\Copyright{Libor Barto, Diego Battistelli, and Kevin M. Berg}



\funding{All three authors have received funding from the European Research Council (ERC) under the
European Unions Horizon 2020 research and innovation programme (grant agreement No 771005).}

\nolinenumbers 

\hideLIPIcs  

\EventEditors{John Q. Open and Joan R. Access}
\EventNoEds{2}
\EventLongTitle{42nd Conference on Very Important Topics (CVIT 2016)}
\EventShortTitle{CVIT 2016}
\EventAcronym{CVIT}
\EventYear{2016}
\EventDate{December 24--27, 2016}
\EventLocation{Little Whinging, United Kingdom}
\EventLogo{}
\SeriesVolume{42}
\ArticleNo{23}

\newtheorem{conjecture}[theorem]{Conjecture}
\theoremstyle{claimstyle}

\newcommand{\rel}[1]{\mathbf{#1}}
\newcommand{\vc}[1]{\mathbf{#1}}
\newcommand{\CSP}{\mathrm{CSP}}
\newcommand{\PCSP}{\mathrm{PCSP}}
\newcommand{\ar}{\mathrm{ar}}
\newcommand{\sel}{\mathrm{sel}}
\newcommand{\sym}{\mathrm{sym}}

\begin{document}

\maketitle
\begin{abstract}
   
The Promise Constraint Satisfaction Problem (PCSP) is a recently introduced vast generalization of the Constraint Satisfaction Problem (CSP). We investigate the computational complexity of a class of PCSPs beyond the most studied cases -- approximation variants of satisfiability and graph coloring problems. We give an almost complete classification for the class of PCSPs of the form: given a 3-uniform hypergraph that has an admissible 2-coloring, find an admissible 3-coloring, where admissibility is given by a ternary symmetric relation. The only PCSP of this sort whose complexity is left open in this work is a natural hypergraph coloring problem, where admissibility is given by the relation ``if two colors are equal, then the remaining one is higher.'' 
\end{abstract}

\section{Introduction}\label{sec:Intro}

The Constraint Satisfaction Problem (CSP) over a finite relational structure $\rel A$ (also called a \emph{template}), denoted $\CSP(\rel A)$, can be defined as a homomorphism problem with a fixed target structure. In the decision version of the problem, an instance is a finite relational structure $\rel X$ (of the same signature as $\rel A$) and the problem is to decide whether there exists a homomorphism (i.e., a relation-preserving map) from $\rel X$ to $\rel A$. In the search version of the problem, we are required to find such a homomorphism whenever it exists. Many computational problems, including various versions of logical satisfiability, graph coloring, and systems of equations can be represented in this form, see the survey~\cite{BKW17}. For example, the CSP over $\rel{K}_3$ (the clique on 3 vertices) is the 3-coloring problem, the CSP over the two-element structure  with all the ternary relations is the 3-SAT problem, the CSP over the two-element structure $\rel{1in3}$ with a single ternary relation $\{(0,0,1),(0,1,0),(1,0,0)\}$ is the positive 1-in-3-SAT, and the CSP over the two-element structure $\rel{NAE}$ with a single ternary relation $\{0,1\}^3 \setminus \{(0,0,0),(1,1,1)\}$ is the positive not-all-equal-3-SAT.

An early general complexity classification result was Schaeffer's dichotomy theorem for Boolean (i.e., two-element) templates~\cite{Sch78}: each such CSP is in P or is NP-complete. Another influential result was a dichotomy theorem by Hell and Ne\v set\v ril~\cite{HN90} for CSPs over (undirected) graphs. Motivated in part by these two theorems, Feder and Vardi formulated their dichotomy conjecture stating that the P versus NP-complete dichotomy remains true for CSPs over arbitrary finite structures. This conjecture inspired a very active research program in the last 20 years, culminating in a celebrated positive resolution independently obtained by Bulatov~\cite{Bul17} and Zhuk~\cite{Zhu17}. Their proofs rely heavily on a general theory of CSPs developed in~\cite{Jea98,BJK05} (the so-called \emph{algebraic approach}) whose central concept is a \emph{polymorphism}, a multivariate function preserving the relations in the template. An NP-hardness criterion in terms of polymorphisms has been isolated in~\cite{BJK05} and was conjectured to be the only source of hardness -- all the templates that do not satisfy this criterion should be in P. This is what Bulatov and Zhuk confirmed in their work. 

This paper studies a recently introduced~\cite{AGH17,BG18} vast generalization of the fixed-template CSP, the \emph{Promise CSP} (PCSP). A promise template is a pair $(\rel A, \rel B)$ of finite relational structures such that $\rel A \to \rel B$, i.e., there exists a homomorphism from $\rel A$ to $\rel B$. The PCSP over $(\rel A,\rel B)$  is then (in the search version) the following problem: given a relational structure $\rel X$ such that $\rel X \to \rel A$ (this is the promise), find a homomorphism $\rel X \to \rel B$ (which is guaranteed to exist since $\rel A \to \rel B$). This framework generalizes that of CSP (where $\rel A = \rel B$) and also includes many important problems in approximation, e.g., if $\rel A = \rel K_{k}$ and $\rel B = \rel K_l$, $k \leq l$, then $\PCSP(\rel A,\rel B)$ is the problem of finding an $l$-coloring of a $k$-colorable graph, a problem whose complexity is open after more than 40 years of research. On the other hand, the basics of the algebraic approach to CSPs from~\cite{BOP18} can be generalized to PCSPs~\cite{BKO19,BBKO}, and this discovery gives us hope that a full complexity classification might be possible. 

Motivated by this ambitious goal, a line of research focuses on studying restricted classes of templates, the two main directions being graph templates and Boolean templates, mimicking the two ``base cases'' in the CSP. For the graph templates, a complexity conjecture has been formulated by Brakensiek and Guruswami in \cite{BG18} stating that all templates that are not in P for simple reasons ($\rel A$ or $\rel B$ has a loop or is bipartite) are NP-hard. Boolean templates form a rich source of examples and are the context where the PCSP framework was introduced~\cite{AGH17,BG18}. Boolean PCSPs are interesting both from the NP-hardness perspective and, unlike the graph templates, from the algorithmic perspective: a generalization of~\cite{BBKO} in \cite{BWZ20} brought the strongest NP-hardness criterion we currently have, which we will also employ in this paper, and the most general polynomial algorithm~\cite{BG20,BGWZ} also came from investigating Boolean PCSPs. The strongest classification result obtained so far in this direction is the dichotomy theorem over Boolean \emph{symmetric} templates, i.e., templates whose relations are all invariant under permutations of coordinates~\cite{BG18,FKOS19}.

Note that both undirected graphs and symmetric Boolean templates use symmetric relations -- undirected graphs consist of a single binary relation over a domain of an arbitrary size, symmetric Boolean templates consist of arbitrarily many symmetric relations of arbitrary arities over a two-element domain. The focus of this paper is somewhere in between these two extremes: we study templates $(\rel A, \rel B)$ where $\rel A$  consists of a single symmetric ternary relation over a two-element set. Our motivation was that the classification of more concrete templates of PCSPs is needed for improving the general theory, such as finding hardness and tractability criteria. The class we study is amenable for case analysis and, on the other hand, already includes important problems both from a hardness and an algorithmic perspective (e.g., $k$-coloring a $2$-colorable 3-uniform hypergraph or finding a positive not-all-equal solution to a satisfiable positive 1-in-3-sat instance).

Let $(\rel A, \rel B)$ be a PCSP template such that $\rel A$ has domain $\{0,1\}$ and a single symmetric ternary relation $R \subseteq \{0,1\}^3$, and let $\rel B$ consist of a single relation $R \subseteq B^3$ on a domain $B$. By taking into account the main result of~\cite{DRS05} and ruling out some trivial cases (see Section~\ref{sec:PandNP} for details), we are left with the case where $\rel A = \rel{1in3}$ and $R$ is also symmetric. 

Note that in this situation $\PCSP(\rel A,\rel B)$ has a hypergraph coloring interpretation: given a 3-uniform hypergraph that is $\rel{1in3}$-colorable (that is, each vertex can be assigned a color from \{0,1\} so that there is exactly one 1 appearing in each hyperedge), find a $\rel B$-coloring (that is, a coloring by $B$ such that the three colors appearing in each hyperedge are from $R$).%
\footnote{We commit a slight imprecision here, since the relation of the instance can contain entries with repeated coordinates, and thus not all instances correspond to 3-uniform hypergraphs. The difference is insignificant for our results.}

\subsection{Three-element domain}

The first non-Boolean domain size, $\vert B \vert =3$, already turns out to be interesting. Figure~\ref{fig:lattice} shows all the templates $\rel B$ ordered by the relation $\rel B \leq \rel B'$ if $\rel B \to \rel B'$ (if $\rel B \leq \rel B' \leq \rel B$, only one representative is drawn). As $\rel B \leq \rel B'$ implies that $\PCSP(\rel{1in3},\rel B')$ reduces to $\PCSP(\rel{1in3},\rel B)$, the higher the template is, the ``easier'' the corresponding PCSP is.
This fact and the notation of the templates are detailed in Sections~\ref{sec:prelim} and \ref{sec:PandNP}.

    \begin{figure}[t]
        \centering
        \includegraphics[width=0.5\textwidth]{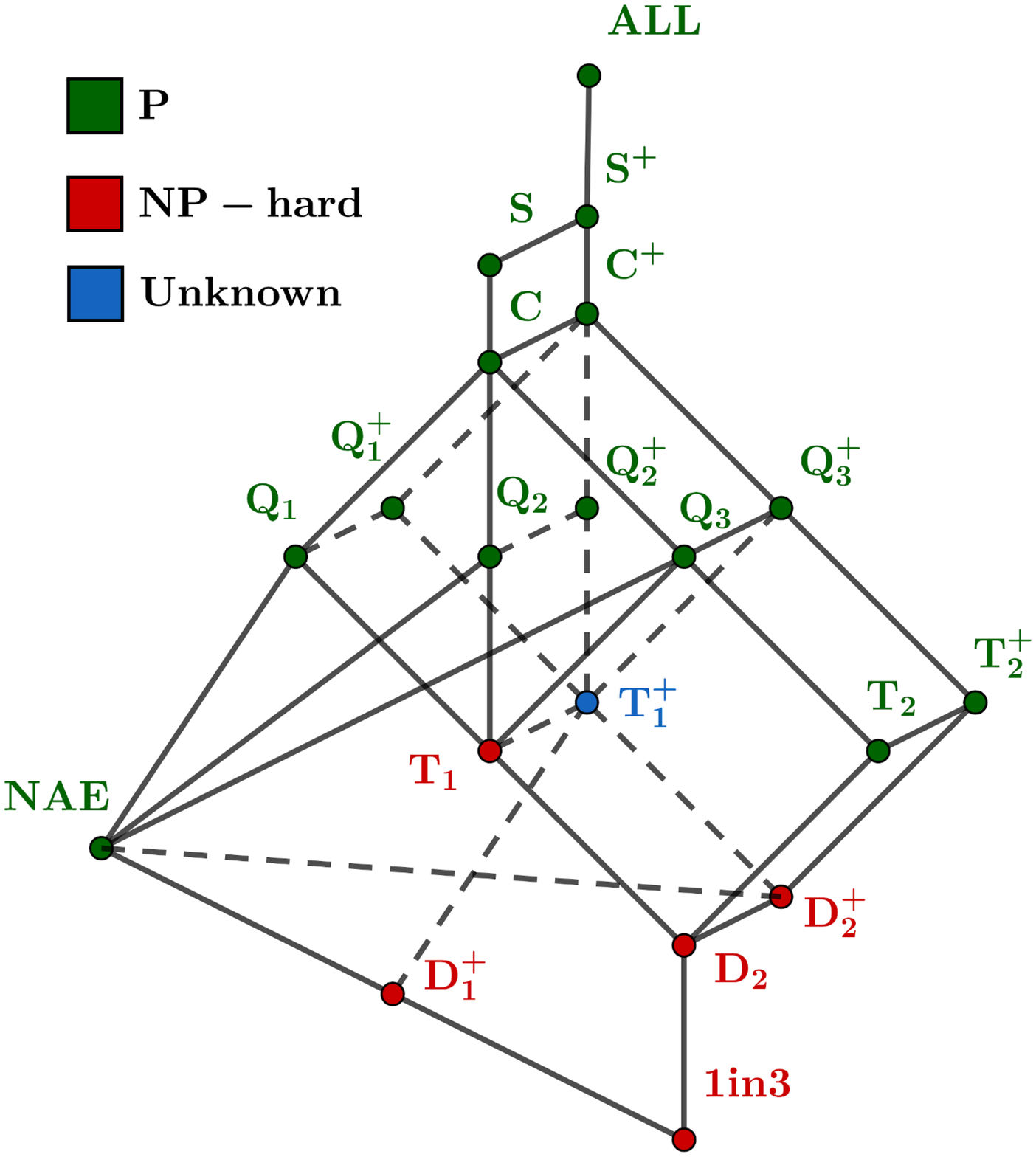}
        \caption{The Lattice of Homomorphism Classes of Symmetric Structures}
        \label{fig:lattice}
    \end{figure}

We were able to classify all but one case:

\begin{theorem}\label{thm:main}
Let $(\rel{1in3},\rel B)$ be a PCSP template, where $\rel B$ has domain-size three.
\begin{itemize}
    \item If $\rel{NAE} \to \rel B$ or $\rel T_2 \to \rel B$, then $\PCSP(\rel{1in3},\rel B)$ is in P.
    \item If $\rel B \to \rel T_1$ or $\rel B \to \rel D_1^+$ or $\rel B \to \rel D_2^+$, then $\PCSP(\rel{1in3},\rel B)$ is NP-hard.
\end{itemize}
\end{theorem}

Even though $|B|=3$ is a small domain size, many interesting phenomena already show up, and we believe that the collection of templates is a valuable source of examples for further exploration. We now emphasize some of the phenomena and open questions.

Tractability of $\rel{NAE}$ and $\rel{T}_2$ can be obtained by using the sufficient condition from~\cite{BG20,BGWZ} -- the existence of block-symmetric polymorphisms of arbitarily large block-sizes. However, $\rel T_2$ is ``simpler'' in that one can use a finite-template CSP to solve it in polynomial-time, while for $\rel{NAE}$ no such finite template exists~\cite{Bar19}. Among the templates in Figure~\ref{fig:lattice}, where is the borderline for such a ``finite tractability''? 

Two features of our NP-hardness results are worth noting. First, they are obtained by applying the currently strongest NP-hardness condition from~\cite{BWZ20}. It seems (but we do not prove it here) that weaker conditions, which were sufficient for NP-hardness of symmetric Boolean CSPs~\cite{BG18,FKOS19} and NP-hardness of approximate hypergraph coloring~\cite{DRS05} (cf.~\cite{BBKO}), are not sufficient here, namely for $\rel{D}_2^+$. The only other such templates we are aware of are those from~\cite{BWZ20}. Second, NP-hardness of $\rel{D}_1^+$ has, similarly to~\cite{DRS05}, topological sources since it employs the high chromatic number of Kneser's graphs~\cite{Lov78} -- a result of Lov\'asz that started topological combinators. However, the application of Kneser's graphs is more direct in our situation, and this may help in further improving the topological methods in PCSPs. In particular, it would be desirable to find a common generalization of the proof in \cite{DRS05} and in the paper~\cite{KO19}, which employs algebraic topology in a seemingly different way. 

The unresolved case, $\rel{T}_1^+$, in fact corresponds to a natural hypergraph coloring problem that appears to be new: given a $\rel{1in3}$-colorable 3-uniform hypergraph, find a 3-coloring such that, in each hyperedge, if two colors are equal, then the third one is \emph{higher} (as opposed to ``different'' for the standard hypergraph coloring). We conjecture that this problem, as well as the natural generalization to larger domains, is NP-complete. If true, there is a unique source of hardness for our templates.

\subsection{Larger domains}

For a 4-element $B$, the above conjecture would resolve all the cases with the exception of the interval between $\rel{\check{C}}$ and $\rel{\check{C}}^+)$, where $\rel{\check{C}}$ is given by the relation containing the tuples  $(0,0,1)$, $(1,1,2)$, $(2,2,3)$, $(3,3,0)$ and their permutations,  and $\rel{\check{C}}^+$ is given by the same relation with all the ``rainbow'' tuples $(i,j,k)$ such that $|\{i,j,k\}| = 3$.%
\footnote{The notation is derived from the Czech word for a square -- \v ctverec.}

We are able to show that the bottom of the interval corresponds to an NP-hard PCSP, and the top one gives a template that does not satisfy the sufficient condition for tractability from~\cite{BG20,BGWZ}.

\begin{theorem}\label{thm:largerdomains}
$\PCSP(\rel{1in3},\rel{\check{C}})$ is NP-hard. The template $(\rel{1in3},\rel{\check{C}}^+)$ does not have a block symmetric polymorphism with two blocks of sizes 23 and 24. 
\end{theorem}

The theorem suggests that even for $|B|=4$, essentially the only tractable cases are $\rel{NAE}$ and $\rel{T}_2$. Is this the case on arbitrary domains?

\section{Basic definitions} \label{sec:prelim}
    
Throughout this paper, we adopt the convention that $[n]=\{1,2, \ldots , n\}$.

A \emph{relational structure} (of a finite signature) is a tuple $\rel A = (A; R_1, \dots, R_n)$, where $A$ is a set called the \emph{domain} of $\rel A$, and each $R_i$ is a relation of arity $\ar_i \geq 1$, that is, a nonempty subset of  $A^{\ar_i}$. A relational structure is \emph{symmetric} if each relation in it is invariant under any permutation of coordinates. 
Two relational structures $\rel A = (A; R_1, \dots, R_n)$ and $\rel B = (B; R'_1, \dots, R'_{n'})$ \emph{have the same signature} if $n=n'$ and each $R_i$ has the same arity as $R'_i$. In this situation, a mapping $f: A \to B$ is a homomorphism from $\rel A$ to $\rel B$, written $f: \rel A \to \rel B$, if it preserves the relations, that is, for each $i$ and each tuple $\vc{a} \in R_i$, we have $f(\vc{a}) \in R_i'$, where $f$ is applied to $\vc{a}$ component-wise. The fact that there exists a homomorphism from $\rel A$ to $\rel B$ is denoted by $\rel A \to \rel B$. 

\begin{definition}
    A \emph{PCSP template} is a pair of finite relational structures with the same signature, $\rel A$, $\rel B$ such that $\rel A \to \rel B$. We denote the PCSP template of $\rel A$ and $\rel B$ by $(\rel A, \rel B)$.
\end{definition}

For a given PCSP template, it is possible to define both a decision problem and a search problem variant of the PCSP.

\begin{definition}
    Let  $(\rel A, \rel B)$ be a PCSP template. The \emph{decision version of $\PCSP(\rel A, \rel B)$} is, given an input structure $\rel I$ with the same signature as $\rel A$ and $\rel B$, to output yes if $\rel I \to \rel A$ and no if $\rel I \not\to \rel B$.

The \emph{search version of $\PCSP(\rel A, \rel B)$} is, given an input structure $\rel I$ with the same signature as $\rel A$ and $\rel B$, to find a homomorphism $h: \rel I \to \rel B$.
\end{definition}

It is not hard to see that the decision version of $\PCSP(\rel A,\rel B)$ can be reduced to its search version. The tractability results in this paper apply to the search version (and hence also to the decision version), while NP-hardness results apply to the decision version (and hence also to the search version).

The following concept captures the situation when one PCSP can be reduced to another one by the trivial reduction, that is, the reduction that does not change the instance.

\begin{definition}
Let $(\rel A,\rel B)$ and $(\rel A',\rel B')$ be PCSP templates of the same signature. We say that $(\rel A',\rel B')$ is a homomorphic relaxation of $(\rel A,\rel B)$ if $\rel A' \to \rel A$ and $\rel B \to \rel B'$. 
\end{definition}

Observe that, indeed, the trivial reduction from $\PCSP(\rel A',\rel B')$ to $\PCSP(\rel A,\rel B)$ is correct if and only if $(\rel A',\rel B')$ is a homomorphic relaxation of $(\rel A,\rel B)$.

A crucial notion for the algebraic approach to PCSP is a \emph{polymorphism}. A polymorphism of a template is simply a homomorphism from a Cartesian power of the first structure to the second one. This can be spelled out as follows.

\begin{definition}
    Let $(\rel{A},\rel{B})$ be a PCSP template. A mapping $f: A^n \to B$ (where $n$ is a positive integer) is a \emph{polymorphism of arity $n$} if, for each pair of corresponding relations $R_i$ and $R'_i$ in the signatures of $\rel A$ and $\rel B$, respectively, and any $(r_{1,1},$ $r_{2,1},$ $\ldots ,$  $r_{n,1}),$ $\ldots , $ $(r_{1,\ar_i},$ $r_{2,\ar_i},$ $\ldots ,$ $r_{n,\ar_i})$ with $(r_{j,1},$ $r_{j,2},$ $\ldots ,$ $r_{j,\ar_i})$ $\in R_i$ for all $j \in [n]$, we have $(f(r_{1,1},$ $r_{2,1},$ $\ldots ,$ $r_{n,1}),$ $\ldots ,$ $f(r_{1,\ar_i},$ $r_{2,\ar_i},$ $\ldots $, $r_{n,\ar_i})) \in R'_i$.
\end{definition}

Another core concept in the algebraic approach is a \emph{minor}.

\begin{definition}
Let $f: A^n \to B$, $\alpha:[n] \to [m]$ be mappings. A \emph{minor} of $f$ given by $\alpha$ is the mapping $f^{\alpha}:A^m \to B$ defined by
\[f^{\alpha}(a_1, \dots, a_m) = f(a_{\alpha(1)}, \dots, a_{\alpha(n)})\]
for every $a_1, \dots, a_m \in A$. A function $g: A^m \to B$ is a minor of $f$ if $g=f^{\alpha}$ for some $\alpha$. 
\end{definition}

The significance of polymorphisms stems from the fact that the computational complexity of $\PCSP(\rel A,\rel B)$ depends only on the set of all polymorphisms of the template $(\rel A,\rel B)$~\cite{BG18,BKO19,BBKO}. 
This set is a \emph{minion}, i.e., it is closed under taking minors.

\section{Templates} \label{sec:templates}

In this section we introduce the notation for the templates considered in the paper, and provide several easy observations about these templates.

We consider symmetric relational structures with a single ternary relation. To each such structure $\rel B = (B; R)$ we associate \emph{its digraph} by taking $B$ as the vertex set and including the arc $b \to b'$ if and only if $(b,b,b') \in R$. By $\rel{B}^+$ we denote the structure obtained from $\rel B$ by adding to $R$ all the tuples $(b,b',b'')$ with $|\{b,b',b''\}|=3$. Note that this is the ``largest'' structure with the same associated digraph as $\rel B$. Also observe that over a three-element domain, i.e., $|B|=3$, there are exactly two structures with the same associated digraph. The notational convention for 3-element structures in Figure~\ref{fig:lattice} is given in Table~\ref{tab:names}.\footnote{The notation is derived from the number of edges of the associated digraph in Italian.} The names in the table refer to the smaller of the two structures with the same digraph, e.g., the relation of $\rel{D}_2$ consists of all permutations of the triples $(0,0,1)$, $(1,1,2)$ while $\rel{D}_2^+$ also contains $(0,1,2)$ and its permutations. Of course, the structure depends on the concrete choice of vertices, but the choice is irrelevant for our purposes.

    \begin{table}[t]
        \centering
 		\caption{Diagrams of Symmetric Structures}
		\label{tab:names}
		\begin{tabular}{|c|c|c|c|c|c|c|}
		    \hline
		    &&&&&&\\
		    Diagram &
		    \begin{tikzpicture}
			    \draw[->](0,0)--(.5,0);
			 \end{tikzpicture} & 
			 \begin{tikzpicture}
			    \draw[->](0,0)--(.5,0);
			    \draw[->](.5,-.1)--(0,-.1);
			 \end{tikzpicture} &
			 \begin{tikzpicture}
			    \draw[->](0,0)--(.5,0);
			    \draw[->](0,0)--(0,.5);
			\end{tikzpicture} &
			\begin{tikzpicture}
			    \draw[->](0,0)--(.5,0);
			    \draw[->](.6,0)--(1.1,0);
			\end{tikzpicture} &
			\begin{tikzpicture}
			    \draw[->](0,0)--(.5,0);
			    \draw[->](0,0)--(0,.5);
			    \draw[->](.6,0)--(.1,.5);
			\end{tikzpicture} &
			 \begin{tikzpicture}
			    \draw[->](.1,0)--(.6,0);
			    \draw[->](0,.5)--(0,0);
			    \draw[->](.7,0)--(.1,.5);
			\end{tikzpicture}\\
			&&&&&&\\
		    \hline 
		    Structure $\mathbf{B}$ & $\mathbf{1in3}$ & $\mathbf{NAE}$ & $\mathbf{D_1}$ & $\mathbf{D_2}$
		    & $\mathbf{T_1}$ & $\mathbf{T_2}$\\
		    \hline
		    \hline
		    &&&&&&\\
		    Diagram &
		    \begin{tikzpicture}
			    \draw[->](0,0)--(.5,0);
			    \draw[->](0,0)--(0,.5);
			    \draw[->](.1,.5)--(.6,0);
			    \draw[->](.7,.1)--(.2,.6);
			\end{tikzpicture} & 
			\begin{tikzpicture}
			    \draw[->](0,0)--(.5,0);
			    \draw[->](.5,-.1)--(0,-.1);
			    \draw[->](0,0)--(0,.5);
			    \draw[->](.1,.5)--(.6,0);
			\end{tikzpicture} &
			 \begin{tikzpicture}
			    \draw[->](0,0)--(.5,0);
			    \draw[->](0,0)--(0,.5);
			    \draw[->](-.1,.5)--(-.1,0);
			    \draw[->](.1,.5)--(.6,0);
			\end{tikzpicture} &
			\begin{tikzpicture}
			    \draw[->](0,0)--(.5,0);
			    \draw[->](0,0)--(0,.5);
			    \draw[->](.5,-.1)--(0,-.1);
			    \draw[->](-.1,.5)--(-.1,0);
			    \draw[->](.1,.5)--(.6,0);
			\end{tikzpicture} &
			\begin{tikzpicture}
			    \draw[->](0,0)--(.5,0);
			    \draw[->](0,0)--(0,.5);
			    \draw[->](.5,-.1)--(0,-.1);
			    \draw[->](-.1,.5)--(-.1,0);
			    \draw[->](.1,.5)--(.6,0);
			    \draw[->](.7,.1)--(.2,.6);
			\end{tikzpicture} &\\
			&&&&&&\\
		    \hline 
		    Structure $\mathbf{B}$ & $\mathbf{Q_1}$ & $\mathbf{Q_2}$ & $\mathbf{Q_3}$ & $\mathbf{C}$
		    & $\mathbf{S}$ &\\
		    \hline
 		\end{tabular}
	\end{table} 

It is a simple exercise to verify that Figure~\ref{fig:lattice} is correct, and we do not give the details here. Let us just observe that $\rel{T}_1^+$ is the only case not covered by Theorem~\ref{thm:main}.
Indeed, the digraph associated to a three-element $\rel B$  either contains a directed cycle, or is acyclic. In the former case, depending on the length of the cycle we have $\rel{NAE} \to \rel{B}$ (length 1 or 2) or $\rel{T}_2 \to \rel{B}$ (length 1 or 3). In the latter case, the digraph can be extended to a linear order, so $\rel B \to \rel{T}_1^+$.
If $\rel B \neq \rel{T}_1^+$, then $\rel B$ has a homomorphism to a symmetric substructure of $\rel{T}_1^+$ with one of the  four triples, $(0,0,1)$, $(0,0,2)$, $(1,1,2)$, $(0,1,2)$, missing. By omitting the second, the third, or the fourth tuple we get the structures $\rel{D}_2^+$, $\rel{D}_1^+$, $\rel{T}_1$ from the second item of Theorem~\ref{thm:main}, respectively. By omitting the first tuple, we get a structure $\rel B$ such that $\rel{B} \to \rel{1in3} \to \rel{B}$, and so this structure sits at the bottom of Figure~\ref{fig:lattice}.

The ``hard'' structure $\rel{T}_1^+$ has a simple description. A tuple $(b_1,b_2,b_3)$ is in its relation if and only if the following condition is satisfied: if two of $b_1,b_2,b_3$ are equal to $b$, then the remaining one must be strictly greater than $b$ in the linear order $0<1<2$. We denote the structure obtained by the same definition on a $k$-element domain ordered $0 < 1 < 2 < \dots < k-1$ by $\rel{LO}_k$, e.g., $\rel{LO}_2 = \rel{1in3}$ and $\rel{LO}_3 = \rel{T}_1^+$.

For the case $|B|=4$, a similar case analysis shows that the only structures $\rel B$ with $\rel{NAE},\rel{T}_2 \not\to \rel{B}$ and $\rel B \not\to \rel{LO}_4$ are the structures whose associated digraph is the directed cycle of length 4 -- the structures in the interval between $\rel{\check{C}}$ and $\rel{\check{C}}^+$ alluded to in the introduction.

Finally, we denote by $\rel{NAE}_k$ the structure with a $k$-element domain and the ternary non-all-equal relation, e.g., $\rel{NAE}_2 = \rel{NAE}$.

\section{Tractability and hardness} \label{sec:PandNP}

In this section we deal with the simple cases, and cases that are resolved by known results. We also provide the hardness criterion that we will employ for the more complex cases.
Throughout the section we consider a PCSP template $(\rel A,\rel B)$ such that $\rel A$ is a relational structure with the two-element domain $\{0,1\}$ and a single symmetric ternary relation, and $\rel B = (B; R)$, where $B = \{0,1, \dots\}$ and $R \subseteq B^3$. 

\subsection{Symmetrization}

First, observe that $R$ can be assumed symmetric without loss of generality. Indeed, for an instance $\rel X$ of $\PCSP(\rel A,\rel B)$ such that $\rel X \to \rel A$, we also have $\rel X^{\sym} \to \rel A$, where $\rel X^{\sym} = (B,R^{\sym})$ is the \emph{symmetrization} of $\rel X$, i.e., $R^{\sym}$ contains all the permutations of the tuples in $R$. It is also easily seen that $\rel X^{\sym} \to \rel B^{\sym}$ implies $\rel X \to \rel B$. It follows that $\rel X \mapsto \rel X^{\sym}$ is a correct reduction from $\PCSP(\rel A, \rel B)$ to $\PCSP(\rel A, \rel B^{\sym})$. In the other direction, we can use the trivial reduction, and therefore these PCSPs are equivalent. 

For the remainder of the section we assume that $\rel B$ is a symmetric structure.

\subsection{Tractability}

If the relation in $\rel A$ or $R$ contains a constant tuple, $(\rel A, \rel B)$ is a homomorphic relaxation of the ``trivial'' template whose two structures have a one-element domain. In particular, $\PCSP(\rel A, \rel B)$ is in P. 

If $\rel A = \rel{1in3}$ and $\rel{NAE} \to \rel B$, then $(\rel A, \rel B)$ is a homomorphic relaxation of $(\rel{1in3},\rel{NAE})$. The PCSP over the latter template is in P by~\cite{BG18}, and therefore so is the PCSP over the former.

The remaining tractable case in Theorem~\ref{thm:main} is  $\rel A = \rel{1in3}$ and $\rel{T}_2 \to \rel B$. These templates are homomorphic relaxations of $(\rel{T}_2,\rel{T}_2)$. But the PCSP over $(\rel{T}_2,\rel{T}_2)$, i.e., the CSP over $\rel{T}_2$, is in P because the relation of $\rel{T}_2$ can be described as $\{(x,y,z) \in \{0,1,2\}^3: x+ y+ z = 1 \pmod{3}\}$, and so $\CSP(\rel{T}_2)$ can be efficiently solved by solving a system of linear equations over the three-element field. 

These tractability results can be also derived from a recent theorem that we now state. We require a definition. A mapping $f: A^n \to B$ is \emph{block-symmetric of width $k$} if there exists a partition of the coordinates of $f$ into blocks $X_1 \cup \dots \cup X_l = [n]$ of size at least $k$ such that $f$ is permutation-invariant within each coordinate block $X_i$.

\begin{theorem}[\cite{BG20,BGWZ}]
The following are equivalent for every PCSP template $(\rel A,\rel B)$.
\begin{itemize}
    \item $(\rel A,\rel B)$ has block-symmetric polymorphisms of arbitrarily high width.
    \item For every $k \in \mathbb{N}$, $(\rel A,\rel B)$ has a block-symmetric polymorphism of arity $2k+1$ with two blocks of size $k$ and $k+1$.
\end{itemize}
If these equivalent conditions are satisfied, then $\PCSP(\rel A, \rel B)$ is in P.
\end{theorem}

In fact, this theorem is strong enough to prove the tractability of all the currently known tractable Boolean PCSPs. In Appendix~\ref{sec:chplus} we use this fact to provide evidence for the NP-hardness of  $\PCSP(\rel{1in3},\rel{\check{C}^+})$: we prove that the template does not have a block-symmetric polymorphism with two blocks of sizes 23 and 24, as claimed in the second part of Theorem~\ref{thm:largerdomains}. 

\subsection{Hardness}

If $\rel A = \rel{NAE}$ and $R$ does not contain a constant tuple, then $(\rel{NAE},$ $\rel{NAE}_{\vert B \vert})$ is a homomorphic relaxation of $(\rel A, \rel B)$, and is therefore NP-hard by the following theorem. 

\begin{theorem}[\cite{DRS05}] \label{thm:hypergraphs}
$\PCSP(\rel{NAE}_k, \rel{NAE}_l)$ is NP-hard for every $2 \leq k \leq l$.
\end{theorem}

The hard cases in Theorems~\ref{thm:main}, ~\ref{thm:largerdomains} -- that is, $\rel A = \rel{1in3}$, and $\rel B = \rel{D}_1^+$, $\rel{D}_2^+$, $\rel{T}_1$, or $\rel{\check{C}}$ -- are dealt with in Sections~\ref{sec:D1plus}, ~\ref{sec:D2plus}, and ~\ref{sec:T1}, and Appendix ~\ref{sec:ch}, respectively.
All of these results employ an NP-hardness criterion that we now state. We will require an additional piece of notation. A \emph{chain of minors} is a sequence of the form $(f_0,$ $\alpha_{0,1},$ $f_1,$ $\alpha_{1,2},$ \dots, $\alpha_{l-1,l},$ $f_l)$ where $f_0, \dots, f_l: A^{n_i} \to B$, $\alpha_{i-1,i}: [n_{i-1}] \to [n_{i}]$, and $f_{i-1}^{\alpha_{i-1,i}} = f_{i}$ for every $i \in [l]$. We will then write $\alpha_{i,j}: [n_i] \to [n_j]$ for the composition of $\alpha_{i,i+1}$, $\alpha_{i+1,i+2}$, \dots, $\alpha_{j-1,j}$. Note that $f_i^{\alpha_{i,j}} = f_j$. 

\begin{theorem}[Corollary 4.2. in~\cite{BWZ20}] \label{thm:wrochna}
Let $(\rel A,\rel B)$ be a PCSP template. Suppose there are constants $k, l \in \mathbb{N}$ 
and an assignment of a set of at most $k$ coordinates $\sel(f) \subseteq [\ar(f)]$ to every polymorphism $f$ of $(\rel A, \rel B)$ such that for every chain of minors $(f_0, \alpha_{0,1}, \dots, f_l)$ with each $f_i$ a polymorphism of $(\rel A,\rel B)$, there are $0 \leq i < j \leq l$ such that
$\alpha_{i,j} (\sel(f_i)) \cap \sel(f_j) \neq \emptyset$ (or, equivalently, 
$\sel(f_i) \cap \alpha_{i,j}^{-1}(\sel(f_j)) \neq \emptyset$). 
Then $\PCSP(\rel A,\rel B)$ is NP-hard.
\end{theorem}

The special case $l=1$ of this theorem is sufficient to prove the NP-hardness of all NP-hard symmetric Boolean PCSPs. For Theorem~\ref{thm:hypergraphs}, $l=1$ is not sufficient; however, it can be derived using a still weaker version of Theorem~\ref{thm:wrochna}.%
\footnote{In fact, the proof in \cite{BBKO} is based on~\cite{DRS05} and applies a version which uses a super-constant $k$ (it is enough that, e.g., $k$ is bounded by a polynomial in the logarithm of the arity of $f$). Wrochna~\cite{Wrochna} has shown that this is not necessary. 
We also remark that Theorem~\ref{thm:wrochna} can be also be strengthened to a super-constant values of $k$. 
}
Theorem~\ref{thm:wrochna} in its full power was first used in~\cite{BWZ20} to prove the NP-hardness of certain symmetric non-Boolean CSPs.

\subsection{0-sets, 1-sets, \dots} \label{subsec:smug}

Given a mapping $f: \{0,1\}^n \to B$ (usually an $n$-ary  polymorphism of $(\rel{1in3},\rel B)$) and a subset of coordinates $X \subseteq [n]$, we write $f(X)$ for the the value $f(a_1, \dots, a_n)$ where $a_i = 1$ if $i \in X$ and $a_i=0$ else. We say that $X$ is a \emph{0-set} if $f(X)=0$. 1-sets, 2-sets, etc. are defined similarly. The function $f$ will be always clear from the context. 

Observe that $f: \{0,1\}^n \to B$ is a polymorphism of $(\rel{1in3}, \rel{B})$ if and only if, for every partition of the coordinates of $f$ into three blocks $X \cup Y \cup Z = [n]$, we have $(f(X),f(Y),f(Z)) \in R$. The forward implication of this observation will be applied many times in the proofs, and we simply say, e.g., ``by compatibility of $f(X)$, $f(Y)$ and $f(Z)$,''  ``by compatibility applied to $X$ and $Y$, \dots,'' or ``by compatibility, \dots'' in such situations.

For example, it is common in our templates for the relation $R$  to have no tuple of the form $(2, 2, \dots)$. Therefore, if $X$ and $Y$ are both known to be 2-sets, we would argue that, by compatibility, it must be the case that $X$ and $Y$ are not disjoint. In such cases, we would say, e.g., ``there are no disjoint 2-sets of $f$.''

One useful feature of $i$-sets is that they are closed under preimages within a chain of minors -- that is, if $(f_0,$ $\alpha_{0,1},$ $f_1,$ $\alpha_{1,2},$ \dots, $\alpha_{l-1,l},$ $f_l)$ is a chain of minors for $(\rel{1in3}, \rel B)$ and $X$ is an $i$-set for some polymorphism $f_{j_1}$ in the chain with $0 \leq j_1 \leq l$, then for any $j_2$ with $0 \leq j_2 < j_1$, $\alpha_{j_2,j_1}^{-1}(X)$ is an $i$-set of $f_{j_2}$.

\section{$\rel{D}_1^+$}\label{sec:D1plus}

In this section we prove the NP-hardness of $\PCSP(\rel{1in3},\rel{D}_1^+)$, where $\rel{D}_1^+ = (\{0,1,2\},R)$ and $R$ consists of all the permutations of the tuples $(0,0,1)$, $(0,0,2)$, and $(0,1,2)$. 

Before applying Theorem~\ref{thm:wrochna}, we first derive several properties of polymorphisms of the template. Let us fix any polymorphism $f: \{0,1\}^n \to \{0,1,2\}$ of  $(\rel{1in3},\rel{D}_1^+)$

\begin{lemma} \label{lm:d11}
 There are no two disjoints 1-sets nor 2-sets.
\end{lemma}
\begin{proof}
If  $X$ and $Y$ are $i$-sets for the same $i \in \{1,2\}$, then $(f(X),f(Y),f([n] \setminus (X \cup Y)) \in R$ by compatibility. But $R$ does not contain any tuple of the form $(1,1,\dots)$ or $(2,2,\dots)$, a contradiction.
\end{proof}

The next lemma uses the high chromatic number of Kneser graphs. Recall that the Kneser graph with parameters $n,m$, denoted $\rel{KG}_{n,m}$, is the graph whose vertices are the $m$-element subsets of $[n]$, and where two vertices are adjacent if and only if the two corresponding sets are disjoint.

\begin{theorem}[\cite{Lov78}]~\label{thm:lovasz}
For $n \geq 2m$, there is no coloring of $\rel{KG}_{n,m}$ by strictly less than $n-2m+2$ colors.
\end{theorem}

\begin{lemma} \label{lm:d12}
$f$ has a 1-set or a 2-set of size at most 3.
\end{lemma}
\begin{proof}
 We first assume that $n\geq 2$ and set $m = \lfloor{(n-2)/2}\rfloor$. Since $n-2m+2 \geq 4$, Theorem~\ref{thm:lovasz} implies that the mapping $X \mapsto f(X)$ cannot be a valid coloring of $\rel{KG}_{n,m}$. Therefore, there are two disjoint sets $X$ and $Y$ of size $m$ such that $f(X) = f(Y)$. By compatibility applied to $X$ and $Y$, the set $Z= [n] \setminus (X \cup Y)$ is a 1-set or 2-set. Its size is at most $n - 2m \leq 3$.

In the case $n=1$, $\{1\}$ is itself a 1-set or a 2-set by compatibility applied to $\emptyset$ and $\emptyset)$.
\end{proof}

We are ready to prove the NP-hardness of our template.

\begin{theorem}
$\PCSP(\rel{1in3},\rel{D}_1^+)$ is NP-hard
\end{theorem}
\begin{proof}
We apply Theorem~\ref{thm:wrochna} with $k=3$ and $l=2$. For a polymorphism $f$ of the template, we define $\sel(f)$ as a 1-set or a 2-set of size at most $3$ -- such a set exists by Lemma~\ref{lm:d12} (if both a small 1-set and a small 2-set exist, we select arbitrarily).

Let $(f_0, \alpha_{0,1}, f_1, \alpha_{1,2}, f_2)$ be a chain of minors consisting of polymorphisms. By the pigeonhole principle, there exists $0 \leq i < j \leq 2$ such that $\sel(f_i)$ and $\sel(f_j)$ is an $m$-set for the same $m \in \{1,2\}$. Then $\alpha_{i,j}^{-1}(\sel(f_j))$ is an $m$-set as well and then $\sel(f_i) \cap \alpha_{i,j}^{-1}(\sel(f_j)) \neq \emptyset$ by Lemma~\ref{lm:d11}, as required.  
\end{proof}

\section{$\rel{D}_2^{+}$}\label{sec:D2plus}
    
    In this section we prove the NP-hardness of $\PCSP(\rel{1in3},\rel{D}_2^+)$, where $\rel{D}_2^+ = (\{0,1,2\},R)$ and $R$ consists of all the permutations of the tuples $(0,0,1)$, $(1,1,2)$, and $(0,1,2)$.
    
	Let $f: \{0,1\}^n \to \{0,1,2\}$ be a polymorphism of $(\rel{1in3},\rel{D}_2^+)$. 
	We start with a lemma that concerns unions of $i$-sets.
	
	\begin{lemma} \label{lm:unions}
	Let $X$ and $Y$ be disjoint subsets of $[n]$. 
	\begin{itemize}
	    \item[(a)] If $f(\emptyset)=0$, $f(X) = 0$, and $f(Y) \in \{0,2\}$, then $f(X \cup Y) \in \{0,2\}$
	    \item[(b)] If $f(\emptyset)=0$, $f(X) = 1$, and $f(Y) \in \{0,1\}$, then $f(X \cup Y) = 1$
	    \item[(c)] If $f(\emptyset)=1$, $f(X) = f(Y) = 1$, then $f(X \cup Y) \in \{0, 1\}$
	    \item[(d)] If $f(\emptyset)=1$, $f(X) = f(Y) = 0$, then $f(X \cup Y) = 2$
	\end{itemize}
	\end{lemma}
	
	\begin{proof}
	 For the first item, by  compatibility applied to $X$ and $Y$, the complement $Z = [n] \setminus (X \cup Y)$ is a 1-set. Therefore, by compatibility applied to $\emptyset$ and $Z$, $X \cup Y$ is a 0-set or a 2-set. The proof for the remaining items is similar.
	\end{proof}
	
	The following lemma will be useful in the case that $\emptyset$ is a 0-set.  Note that in this case $[n]$ is a 1-set by compatibility applied to $\emptyset$ and $\emptyset$.

	\begin{lemma}\label{lm:d2l1} Assume $f(\emptyset) = 0$ and that $f$ has no singleton 2-set. Then $f$ has  a singleton 1-set and does not have any two disjoint 1-sets. 
	\end{lemma}
	
	\begin{proof}
	If every singleton is a 0-set, then by adding to $\emptyset$ singletons, one by one, and using item (a) of Lemma~\ref{lm:unions}, we get that $[n]$ is a 0-set or a 2-set. This contradiction shows that there exists a singleton 1-set. 
		
	For the second part of the claim, suppose $X$ and $Y$ are disjoint 1-sets. By adding to $Y$ singletons and using item (b) of Lemma~\ref{lm:unions}, we obtain that $[n] \setminus X$ is a 1-set, a contradiction to compatibility applied to $\emptyset$ and $X$. 
	\end{proof}

	We now consider the case that $\emptyset$ is a 1-set.  Observe that $[n]$ is a 2-set in this case.

	\begin{lemma}\label{lemma:successor} Assume $f(\emptyset)=1$ and that, for some $j \geq 2$, all at most $j$-element subsets of $[n]$ are 1-sets. Then $j < n$ and all $(j+1)$-element subsets sets of $[n]$ are 1-sets. 
	\end{lemma}
	
	\begin{proof}
	Clearly $j < n$ as $[n]$ is a 2-set. Assume, for a contradiction, that for some $j$-element $X$ and $y \not\in X$, the set $Y := X \cup \{y\}$ is not a 1-set. Since $X$ and $\{y\}$ are 1-sets, then $Y$ is a 0-set by item (c) of Lemma~\ref{lm:unions}.
	
	We prove by induction on $i$ that every set $Z$ of size $i$ disjoint with $Y$ is a 1-set. The base case of the induction may be, e.g., $i=0$ (or $i=j$). For the induction step, consider an $(i+1)$-element $Z$ disjoint from $Y$ and write $Z = Z' \cup \{z\}$ where $|Z'|=i$. By the induction hypothesis $Z'$ is a 1-set. The set $\{y,z\}$ is a 1-set as well by assumption (note that $j \geq 2$). Therefore $Z' \cup \{y,z\} = Z \cup \{y\}$ is a 0-set or 1-set by item (c) of Lemma~\ref{lm:unions}. By compatibility applied to $X$ and $Z \cup \{y\}$, the complement $W = [n] \setminus (X \cup Z \cup \{y\}) = [n] \setminus (Y \cup Z)$ is a 0-set or 2-set. But then, by compatibility applied to $Y$ (a 0-set) and $W$, $Z$ is a 1-set, as required.
	
	We have proved that $[n] \setminus Y$ is a 1-set, a contradiction to compatibility applied to $\emptyset$ and $Y$
	\end{proof}

	\begin{lemma}\label{lm:d2l2} If $f(\emptyset)=1$, then there exists a $0$-set or a $2$-set of size at most 2.  
	\end{lemma}
	
	\begin{proof}
	In the opposite case, every set of coordinates of size at most 2 is a 1-set. It would then follow from Lemma~\ref{lemma:successor} that $[n]$ is a 1-set, a contradiction.
	\end{proof}

	Equipped with these lemmata, we can now proceed to our main argument for this section.
	
	\begin{theorem}
	$\PCSP(\rel{1in3},\rel{D}_2^+)$ is NP-hard.
	\end{theorem}
	
	\begin{proof}
	We apply Theorem~\ref{thm:wrochna} with $k=2$ and $l=5$.  We assign to a polymorphism its \emph{type} and define $\sel(f)$ as follows.
	
	\begin{itemize}
	    \item Type 1: $f$ has a 2-set $X$ of size at most 2. In this case we set $\sel(f) = X$. 
	    \item Type 2: $f$ has no 2-set of size at most 2, $f(\emptyset)=0$, and $\{x\}$ is a 1-set for some $x \in [n]$. We set $\sel(f) = \{x\}$. 
	    \item Type 3: $f$ has no 2-set of size at most 2, $f(\emptyset)=1$, and $f$ has a 0-set $X$ of size at most 2. We set $\sel(f) = X$. 
	\end{itemize}
	Note that $\emptyset$ cannot be a 2-set. The first part of Lemma~\ref{lm:d2l1} and Lemma~\ref{lm:d2l2} then guarantee that every polymorphism is of one of the three types. 
	
    Let $(f_0, \alpha_{0,1}, \dots, f_l)$ be a chain of minors consisting of polymorphisms. Note that $f_i(\emptyset)$ does not depend on $i$, therefore types 2 and 3 do not simultaneously occur in the chain. If, for some $i<j$, both $f_i$ and $f_j$ have type 1, then $\sel(f_i)$ and $\alpha_{i,j}^{-1}(\sel(f_j))$ are both 2-sets, so they have a nonempty intersection. Similarly, if two polymorphisms in this chain have type 2, then we obtain a nonempty intersection by the second part of Lemma~\ref{lm:d2l1}. 
    
    Otherwise, since $l = 5$, the chain contains four polymorphisms $f_{i_1}$, $f_{i_2}$, $f_{i_3}$, $f_{i_4}$ of type 3 (where $i_1 < i_2 < i_3 < i_4$). Let $X_1 = \sel(f_{i_1})$ and $X_j = \alpha_{i_1,i_j}^{-1}(\sel(f_{i_j}))$ for $j = 2,3,4$. 
    These four sets are 0-sets (as preimages of 0-sets). If they are pairwise disjoint, then $X_1 \cup X_2$ and $X_3 \cup X_4$ are disjoint sets, which are 2-sets by item (d) in Lemma~\ref{lm:unions}, a contradiction.
    Therefore, two of these sets, say $X_j$ and $X_{j'}$, have a nonempty intersection. But then $Y := \sel(f_{i_j})$ and $Z := \alpha_{i_j,i_{j'}}^{-1}(\sel(f_{i_j'}))$ also have a nonempty intersection as $X_j = \alpha_{i_1,i_j}^{-1}(Y)$ and $X_{j'} = \alpha_{i_1,i_{j}}^{-1}(Z)$. 
    \footnote{The last part of the argument applies Theorem~\ref{thm:wrochna} in a similar way as in~\cite{BWZ20}, see their ``smug sets'' Corollary 4.2.}
\end{proof}

\section{$\rel{T}_1$}\label{sec:T1}

    In this section we prove the NP-hardness of $\PCSP(\rel{1in3},\rel{T}_1)$, where $\rel{T}_1 = (\{0,1,2\},R)$ and $R$ consists of all the permutations of the tuples $(0,0,1)$, $(0,0,2)$, and $(1,1,2)$.

	Let $f: \{0,1\}^n \to \{0,1,2\}$ be a polymorphism of $(\rel{1in3},\rel{T}_1)$. Note that, as in the previous cases, $f$ cannot have two disjoint 2-sets.
	In particular, $\emptyset$ is a 0-set or a 1-set. 
	The following simple lemma will be useful in both cases.
	
	\begin{lemma} \label{lm:subunion}
	If $Z \subseteq X \cup Y$, and $X$ and $Y$ are 1-sets, then $Z$ is not a 2-set.
	\end{lemma}
	\begin{proof}
	By compatibility applied to $X$ and $Y$, the complement $[n] \setminus (X \cup Y)$ is a 2-set. Since it is disjoint with $Z$, $Z$ cannot be a 2-set. 
	\end{proof}

	For the case $f(\emptyset)=0$ we introduce some notation. 
	We define $r: \{0,1,2\} \to \{0,1\}$ by $0 \mapsto 0$ and $1,2 \mapsto 1$, and set 
	\[
	E(f) = \{x \in [n]: r(f(\{x\}) = 1\}, \mbox{ and }
	I(f) = \{x \in [n]: r(f(\{x\}) = 0\} = [n] \setminus E(f).
	\]
	
	\begin{lemma} \label{lm:linear}
	The size of $E(f)$ is odd and,
	for any $X \subseteq [n]$, we have $r(f(X)) = |X \cap E(f)| \mod 2$. 
	\end{lemma}
	\begin{proof}
	By the ``union argument'' as in the proof of Lemma~\ref{lm:unions}, we get that for any two disjoint $Y$ and $Z$, $r(f(Y \cup Z)) = r(f(Y)) + r(f(Z))$, where the addition is modulo 2. It then follows (by adding to $\emptyset$ singletons from $X$) that $r(f(X)) = \sum_{x \in X} r(f(\{x\})) = |X \cap E(f)| \pmod 2$. 
	
	In particular, $r(f([n])) = |E(f)| \pmod 2$. But $[n]$ cannot be a 0-set (by compatibility applied to $\emptyset$ and $\emptyset$), so $|E(f)|$ is odd.
	\end{proof}
	
	\begin{lemma} \label{lm:addIf}
	Suppose that $f(\emptyset)=0$ and $f$ does not have any 2-sets of size 2. If $X$ is a 1-set such that $E(f) \setminus X$ is nonempty, then $X \cup I(f)$ is a 1-set.
	\end{lemma}
	\begin{proof}
	It is enough to show that $X \cup \{y\}$ is a 1-set for any $y \in I(f)$, as the claim then follows by induction.
	By Lemma~\ref{lm:linear}, $X \cup \{y\}$ is a 1-set or 2-set, so it is enough to exclude the latter option.
	Take $z \in E(f) \setminus X$. By Lemma~\ref{lm:linear}, $\{y,z\}$ is a 1-set or a 2-set, therefore it is a 1-set by the assumption. But then $X \cup \{y\}$ is not a 2-set by Lemma~\ref{lm:subunion}.
	\end{proof}
	
	\begin{lemma} \label{lm:sizes}
	Assume that $f(\emptyset)=0$ and $f$ does not have any singleton 2-set. If $X \subseteq E(f)$ is a 1-set, then, for any $Y \subseteq E(f)$ with $|Y|=|X|$, $Y$ is a 1-set.
	\end{lemma}
	\begin{proof}
	We will show that $Z := X \setminus \{x\} \cup \{y\}$ is a 1-set for any $x \in X$ and $y \in E(f)$. The claim will then follow by induction since any set $Y$ can be obtained from $X$ by a sequence of such ``swaps''. 
	
	By Lemma~\ref{lm:linear} and the non-existence of singleton 2-sets, $|X|$ is odd, $Z$ is a 1-set or 2-set, and $\{y\}$  is a 1-set (it cannot be a 2-set by the assumption of the lemma). Since $Z \subseteq X \cup \{y\}$, then $Z$ is a 1-set by Lemma~\ref{lm:subunion}.
	\end{proof}

	\begin{lemma} \label{lm:smallEf}
	If $f(\emptyset)=0$ and $f$ does not have any 2-sets of size at most 2, then $|E(f)| \leq 5$.
	\end{lemma}
	\begin{proof}
	We first observe that any $X \subseteq E(f)$ of odd size $|X| \leq |E(f)|/2$ is a 1-set. Indeed, otherwise we can find $Y$ disjoint from $X$ of the same size. By Lemma~\ref{lm:linear} and Lemma~\ref{lm:sizes}, both $X$ and $X'$ are 2-sets, a contradiction.
	
	By Lemma~\ref{lm:linear}, the size $i := |E(f)|$ is odd. If $i > 5$, then $E(f)$ can be written as a disjoint union $E(f) = X \cup Y \cup Z$ of sets that have odd sizes smaller than $|E(f)|/2$. By the previous paragraph, all of these sets are 1-sets. But then, by Lemma~\ref{lm:addIf}, $Z \cup I(f) = [n] \setminus (X \cup Y)$ is a 1-set as well, a contradiction to compatibility applied to $X$ and $Y$. 
	\end{proof}
	
	We now consider the case that $\emptyset$ is a 1-set. Observe that $[n]$ is a 2-set in this case.
	
	\begin{lemma} \label{lm:nonidemp}
	If $f(\emptyset)=1$, then $f$ has a 2-set of size at most 2.
	\end{lemma}
	\begin{proof}
	Assume, for a contradiction, that there are no 2-sets of size at most 2. 
	
	Union arguments in the case $f(\emptyset)=1$ give us $r(f(Y \cup Z)) = r(f(X)) + r(f(Y)) + 1 \pmod 2$ and, as in Lemma~\ref{lm:linear}, we obtain that $X = \{x \in [n]: f(\{x\})=0\}$ has an odd size and that, using additionally the ``no two-element 2-sets'' assumption, every two-element subset of $X$ is a 1-set. 
	
	Note that if $Y$ and $Z$ are disjoint 1-sets, then the union argument gives us a sharper result -- $Y \cup Z$ is a 1-set.
	It follows that $X \setminus \{x\}$, where $x \in X$ is an arbitrary element, is a 1-set (as it is a disjoint union of 2-element subsets of $X$) and $[n] \setminus X$ is a 1-set (as it is a disjoint union of singletons outside $X$, which are 1-sets by the ``no singleton 2-set'' assumption). Now compatibility applied to $X \setminus \{x\}$ and $[n] \setminus X$ gives us that $\{x\}$ is a 2-set, a contradiction. 
	\end{proof}

	\begin{theorem}
	$\PCSP(\rel{1in3},\rel{T}_1)$ is NP-hard.
	\end{theorem}
	
	 \begin{proof}
	        
	        We apply Theorem~\ref{thm:wrochna} with $k=5$ and $l=2$.  We assign to a polymorphism its \emph{type} and define $\sel(f)$ as follows.
	        
	        \begin{itemize}
	            \item Type 1: $f$ has a 2-set $X$ of size at most 2. In this case we set $\sel(f) = X$. 
	            
	            \item Type 2: $f$ has no 2-set of size at most 2. In this case we set $\sel(f) = E(f)$. 
	        \end{itemize}
	        
	        Note that $\sel(f)$ in type 1 is nonempty and type 2 only occurs in case that $f(\emptyset)=0$ (by Lemma~\ref{lm:nonidemp}) and then $E(f)$ has size at most 5 by Lemma~\ref{lm:smallEf}.
        	
            Let $(f_0, \alpha_{0,1}, f_1, \alpha_{0,2}, f_2)$ be a chain of minors consisting of polymorphisms. If both $f_i$ and $f_j$ (where $i<j$) have type 1, then $\sel(f_i)$ and $\alpha_{i,j}^{-1}(\sel(f_j))$ are both 2-sets, so they have a nonempty intersection.
                Otherwise, since $l=2$, the chain contains 2 polymorphisms $f_{i}$, $f_{j}$ of type 2 (where $i< j$). We have $f_i(\emptyset)=f_j(\emptyset)=0$ and $\alpha_{i,j}^{-1}(\sel(f_j))$ is a 1-set of $f_i$.
            By Lemma~\ref{lm:linear}, this 1-set has an odd-sized intersection with $E(f_i) = \sel(f_i)$, in particular $\sel(f_i) \cap \alpha_{i,j}^{-1}(\sel(f_j)) \neq \emptyset$.
        \end{proof}

\section{Conclusion}\label{sec:conclusion}

The investigation of PCSPs over the templates $(\rel A,\rel B)$, with $\rel A$ a Boolean structure consisting of a single ternary symmetric relation, boils down to $\PCSP(\rel{1in3}=\rel{LO}_2,\rel B)$ where $\rel B$ is symmetric. We have classified the computational complexity for all such three-element structures $\rel B$ with the exception of $\rel B = \rel{LO}_3$. The remaining case, and its  generalization to larger domains, is a natural computational problem -- recall the interpretation as a version of  hypergraph coloring from the introduction. We conjecture that all of them are NP-hard.

\begin{conjecture}
For every $2 \leq k < l$, $\PCSP(\rel{LO}_k,\rel{LO}_l)$ is NP-hard.
\end{conjecture}

A possible intermediate step to resolving the smallest unknown case, $\PCSP(\rel{LO}_2$, $\rel{LO}_3)$, is to replace $\rel{LO}_2$ by a 3-element structure in the interval between $\rel{1in3}$ and $\rel{LO}_3$.

For  four-element structures, the remaining cases additionally include the structures in the interval between $\rel{\check{C}}$ and $\rel{\check{C}^+}$. We proved NP-hardness for $\rel{\check{C}}$ and provided evidence suggesting that $\rel{\check{C}^+}$ also gives rise to an NP-hard PCSP:

\begin{conjecture}
$\PCSP(\rel{1in3},\rel{\check{C}^+})$ is NP-hard. 
\end{conjecture}

Negative resolution of this conjecture would also be valuable -- it would require a polynomial-time algorithm that has not yet been used for PCSPs.

Observe that a homomorphism from a 3-uniform hypergraph  to $\rel{\check{C}^+}$ also has a nice interpretation, as a coloring by 4-colors such that if two vertices of an hyperedge receive the same color, the last vertex must receive a color which is one higher (mod 4). Other templates admit a natural interpretation and generalizations as well, e.g., $\rel B = \rel{D}_2$.

Finally, it seems possible that $(\rel{1in3},\rel{NAE})$ and $(\rel{1in3},\rel{T}_2)$ are essentially the only tractable templates for arbitrary domain sizes. We do not feel that we have enough evidence supporting such a conjecture, so we refrain from phrasing it.
\bibstyle{plainurl}
\bibliography{pcsp}

\appendix
\section{$\rel{\check{C}}$}\label{sec:ch}
    
    In this appendix we begin the proof of Theorem \ref{thm:largerdomains} by verifying the first statement: the NP-hardness of $\PCSP(\rel{1in3},\rel{\check{C}})$, where $\rel{\check{C}} = (\{0,1,2,3\},R)$ and $R$ consists of all the permutations of the tuples $(0,0,1)$, $(1,1,2)$, $(2,2,3)$, and $(0,3,3)$.
    
    Before applying Theorem~\ref{thm:wrochna}, we begin by deriving several properties of polymorphisms of the template. Let $f: \{0,1\}^n \to \{0,1,2,3\}$ be a polymorphism of $(\rel{1in3},\rel{\check{C}})$.
    
    \begin{lemma}\label{lm:chForbid} Suppose $f(\emptyset) = i$. 
        
        \begin{itemize}
        
            \item[(a)] $f$ has no $(i+2 \mod 4)$-sets. 
            
            \item[(b)] $f$ has no two disjoint $(i+1 \mod 4)$-sets.
            
        \end{itemize}
        
        \begin{proof}
            Suppose $f(\emptyset) = 0$. If $X \subseteq [n]$ were a 2-set, then by compatibility with $\emptyset$ it would be the case that $[n] \setminus X$ is compatible with a 0-set and a 2-set. There are no such sets, proving item (a) for this value. Furthermore, if $X \subseteq [n]$ and $Y \subseteq [n]$ are both disjoint 1-sets, then by compatibility with $X$ and $Y$, $[n] \setminus (X \cup Y)$ is a 2-set, but there are no such sets, proving item (b) for this value.
            
            The proof for the remaining values of $i$ is similar.
        \end{proof}
    \end{lemma}
    
    Union arguments (see the proof of Lemma~\ref{lm:unions}) give us the following properties.
    
    \begin{lemma}\label{lm:chUnion} Let $X$ and $Y$ be disjoint subsets of $[n]$.
        
        \begin{itemize}
        
            \item[(a)] If $f(\emptyset) = f(X) = f(Y) = i$, then $f(X \cup Y) = i$.
            
            \item[(b)] If $f(\emptyset) = i$ and $f(X) = f(Y) = i+3 \mod 4$, then $f(X \cup Y) = i+1 \mod 4$.
            
        \end{itemize}
    \end{lemma}
    
    Finally, we prove a lemma about small $i$-sets which will facilitate our main argument for this appendix.
    
    \begin{lemma}\label{lm:chLimits} If $f(\emptyset) = i$, and $f$ has no $(i+3 \mod 4)$-set with size at most 2, then there exists a
    singleton $(i+1 \mod 4)$-set.
    
        \begin{proof}
        
            We will consider the case where $f(\emptyset) = 0$, as proofs for other values of $i$ will be similar. Observe that in this case, $[n]$ is a 1-set by compatibility applied to $\emptyset$ and $\emptyset$. 
            Suppose by way of contradiction that no such $y \in [n]$ exists. It must then be the case that every singleton is a 0-set. However, by adding to $\emptyset$ singletons, one by one, and using item (a) of Lemma~\ref{lm:chUnion}, we get that $[n]$ is a 0-set, a contradiction.
            
        \end{proof}
    \end{lemma}
    
    Equipped with these lemmata, we can now proceed to our main argument for this appendix.
    
    \begin{theorem}\label{thm:ch} $\PCSP(\rel{1in3},\rel{\check{C}})$ is NP-hard.
        
        \begin{proof}
            We apply Theorem~\ref{thm:wrochna} with $k=2$ and $l=5$.  We assign to a polymorphism with $f(\emptyset) = i$ its \emph{type} and define $\sel(f)$ as follows.
	
	        \begin{itemize}
	            \item Type 1: $f$ has a $(i+3 \mod 4)$-set $X$ of size at most 2. In this case we set $\sel(f) = X$. 
	            \item Type 2: $f$ does not have a $(i+3 \mod 4)$-set of size at most 2 and $f$ has a singleton  $(i+1 \mod 4)$-set $\{x\}$. We set $\sel(f) = \{x\}$.
	        \end{itemize}
	        
	        Lemma~\ref{lm:chLimits} guarantees that every polymorphism is of one of the two types.
	
            Let $(f_0, \alpha_{0,1}, \dots, f_l)$ be a chain of minors consisting of polymorphisms and note that the value at $\emptyset$ is constant throughout the chain. For simplicity, let this value be 0.  If, for some $i<j$, both $f_i$ and $f_j$ have type 2, then $\sel(f_i)$ and $\alpha_{i,j}^{-1}(\sel(f_j))$ are both 1-sets, so they have a nonempty intersection by item (b) of Lemma~\ref{lm:chForbid}.
    
            Otherwise, since $l = 5$, the chain contains four polymorphisms $f_{i_1}$, $f_{i_2}$, $f_{i_3}$, $f_{i_4}$ of type 1 (where $i_1 < i_2 < i_3 < i_4$). Let $X_1 = \sel(f_{i_1})$ and $X_j = \alpha_{i_1,i_j}^{-1}(\sel(f_{i_j}))$ for $j = 2,3,4$. 
            
            These four sets are 3-sets (as preimages of 3-sets). If they are pairwise disjoint, then $X_1 \cup X_2$ and $X_3 \cup X_4$ are disjoint sets, which are 1-sets by item (b) in Lemma~\ref{lm:chUnion}, a contradiction with item (b) of Lemma~\ref{lm:chForbid}.
            
            Therefore, two of these sets, say $X_j$ and $X_{j'}$, have a nonempty intersection. But then $Y := \sel(f_{i_j})$ and $Z := \alpha_{i_j,i_{j'}}^{-1}(\sel(f_{i_j'}))$ also have a nonempty intersection as $X_j = \alpha_{i_1,i_j}^{-1}(Y)$ and $X_{j'} = \alpha_{i_1,i_{j}}^{-1}(Z)$.
        \end{proof}
        
    \end{theorem}

\section{$\rel{\check{C}^+}$}\label{sec:chplus}

    Recall that $\rel{\check{C}^+} = (\{0,1,2,3\},R)$, where $R$ consists of all the permutations of the tuples $(0,0,1)$, $(1,1,2)$, $(2,2,3)$, and $(0,3,3)$, as well as all the ``rainbow'' tuples $(i,j,k)$ such that $|\{i,j,k\}| = 3$. In this appendix we show that there is no block symmetric polymorphism of $(\rel{1in3},\rel{\check{C}^+})$ with two blocks of sizes 23 and 24.
    
        \begin{lemma} If $g: \{0,1\}^{47} \to \{0,1,2,3\}$ is a block symmetric polymorphism of $(\rel{1in3},\rel{\check{C}^+})$ with two blocks of sizes 23 and 24, then there exists a symmetric polymorphism $f: \{0,1\}^{23} \to \{0,1,2,3\}$  of $(\rel{1in3},\rel{\check{C}^+})$.
        \end{lemma}
            
        \begin{proof}
                
            We will define $f: \{0,1\}^{23} \to \{0,1,2,3\}$ based on the symmetric blocks of $g$, which we name $X_{23}$ and $X_{24}$ in accordance with their sizes. For any $X \subseteq [23]$, we set $f(X) = g(Y \cup Z)$, where $Y \subseteq X_{23}$ with $\vert Y \vert = \vert X \vert$ and $Z \subseteq X_{24}$ with $\vert Z \vert = 8$. By construction, then, $f$ is a symmetric polymorphism of $(\rel{1in3},\rel{\check{C}^+})$.
                
        \end{proof}
    
    \begin{theorem} There is no block symmetric polymorphism of $(\rel{1in3},\rel{\check{C}^+})$ with two blocks of sizes 23 and 24.
    \end{theorem}
        
        \begin{proof}
            
            Suppose by way of contradiction that there is such a polymorphism, say $g: \{0,1\}^{47} \to \{0,1,2,3\}$.
            Let $f: \{0,1\}^{23} \to \{0,1,2,3\}$ be a symmetric polymorphism of $(\rel{1in3},\rel{\check{C}^+})$, guaranteed by the previous lemma. 
            
            Since $f$ is symmetric, we adopt the convention that $f(m)$ is the value of $f(X)$ for any $X \subseteq [n]$ with $\vert X \vert = m$. Assume now that $f(8) = 0$ -- our argument will be constructed such that other choices for the value of $f(8)$ can be carried forward to likewise achieve a contradiction. By compatibility with $f(8)$ and $f(8)$, we have then that $f(7) = 1$, and similarly by compatibility with $f(7)$ and $f(8)$ it must be the case that $f(9) = 2$. Since $f(9) = 2$, by compatibility with $f(9)$ and $f(9)$ it follows that $f(5) = 3$. In turn, by compatibility with $f(5)$ and $f(5)$, we get that $f(13) = 0$. Since $f(8) = f(13) = 0$, it must then be the case by compatibility that $f(2) = 1$. Therefore, since $f(14)$ is compatible with $f(7)=1$ and $f(2)=1$, we get that $f(14) = 2$. Since $f(9) = f(14) = 2$, it follows in turn by compatibility that $f(0) = 3$.
            
            Consider now the possible values of $f(6)$. If $f(6) = 0$, then $f(9) = 1$ by compatibility with $f(6)$ and $f(8)$, but it has already been shown that $f(9) = 2$, a contradiction. If $f(6) = 2$, then $f(8) = 3$ by compatibility with $f(6)$ and $f(9)$, but by our initial assumption, $f(8) = 0$, a contradiction. If $f(6) = 1$, then $f(11) = 2$ by compatibility with $f(6)$ and $f(6)$, and $f(10) = 2$ by compatibility with $f(6)$ and $f(7)=1$. However, it must then be the case by compatibility with $f(10)$ and $f(11)$ that $f(2) = 3$, but it has already been shown that $f(2) = 1$, a contradiction. Finally, if $f(6) = 3$, then $f(11) = 0$ by compatibility with $f(6)$ and $f(6)$. Similarly, by compatibility with $f(5)$ and $f(6)$, we get that $f(12) = 0$. But $f(0) = 3$ and is compatible with $f(11) = f(12) = 0$, which is a contradiction since no permutation of $(0,0,3)$ is in $R$. Therefore, no value of $f(6)$ is possible, and thus no such $g$ exists.
            
        \end{proof}

    This theorem, together with Theorem~\ref{thm:ch}, completes the proof of Theorem~\ref{thm:largerdomains}.

\end{document}